\documentclass[12pt]{article}
\usepackage[english]{babel}
\usepackage{amsmath,amsthm,amssymb,amsfonts}
\usepackage{latexsym,xspace,url}
\usepackage{color}

\newtheorem{Theorem}{Theorem}[section]

\newtheorem{Remark}[Theorem]{Remark}
\newtheorem{Definition}[Theorem]{Definition}

\usepackage{float}
\usepackage[colorlinks=true,linkcolor=blue,citecolor=red, urlcolor=green]{hyperref} 
\usepackage{cite}

\begin{document}
\title{Approximate Noether symmetries\\ of perturbed Lagrangians\\
and approximate conservation laws}

\author{M.~Gorgone, F.~Oliveri\\
\ \\
{\footnotesize Department of Mathematical and Computer Sciences,}\\
{\footnotesize Physical Sciences and Earth Sciences, University of Messina}\\
{\footnotesize Viale F. Stagno d'Alcontres 31, 98166 Messina, Italy}\\
{\footnotesize matteo.gorgone@unime.it; francesco.oliveri@unime.it}
}

\date{Published in \emph{Mathematics}, \textbf{9}, 2900, 2021.}

\maketitle

\begin{abstract}
In this paper, within the framework of the consistent approach recently introduced for approximate Lie symmetries of differential equations, we consider approximate Noether symmetries of variational problems involving small terms. 
Then, we state an approximate Noether theorem leading to the construction of approximate conservation laws. Some illustrative applications are~presented.
\end{abstract}

Approximate Noether theorem; Approximate conservation laws


\section{Introduction}
Methods of Lie theory of continuous transformations~\cite{Ovsiannikov1982book,Olver1986book,Ibragimov1994CRC1,Ibragimov1995CRC2,Ibragimov1996CRC3,BlumanAnco2002book,BlumanCheviakovAnco2009book} yield a general framework for deeply investigating both ordinary and partial differential equations. In the case of differential equations deduced from a Lagrangian function through a variational technique~\cite{Marsden-book}, thanks to Noether's theorem~\cite{Noether}, Lie symmetries are intimately connected to conservation laws (first integrals in the case of ordinary differential equations).
 
A lot of models arising in concrete applications write as differential equations involving terms with very small coefficients;  the latter usually have the effect of destroying some important symmetries and so limiting the applicability of Lie group methods. Nevertheless, perturbative techniques are often successful in the investigation of  concrete models that can be viewed as small perturbations of exact models~\cite{Nayfeh}. In such a context, in order to extend the applicability of Lie symmetry methods to such kinds of problems, some approximate symmetry theories 
have been proposed and widely applied in the last decades.

There are two popular approaches to approximate Lie symmetries. One has been proposed by Baikov, Gazizov and Ibragimov (BGI)~\cite{BGI-1989,IbragimovKovalev}, consisting in   
expanding in a perturbation series the Lie generator and imposing the approximate invariance of the equations at hand. A slightly different method has been introduced by Fushchich and Shtelen (FS)~\cite{FS-1989}. In this case, the dependent variables are expanded  in a series as 
done in usual perturbation analysis; terms are then separated at each order of approximation, and a system of equations to be solved in a hierarchy is obtained; this resulting system 
is assumed to be coupled, and the approximate symmetries of the original equations are defined as the 
{exact symmetries} of the approximate equations.
Since their introduction, these two proposals have been applied to many physical models 
(see, for instance, 
\cite{Wiltshire1996,BaikovKordyukova2003,DolapciPakdemirli2004,Wiltshire2006,
GazizovIbragimovLukashchuk2010,GazizovIbragimov2014,Euler1,Euler2,Euler3,Diatta}). The BGI method has also been used for deriving an approximate Noether theorem~\cite{Noether,Govinder-1998,IbragimovUnalJogreus2004}.

Both methods have pros and cons. In the BGI approach, despite the elegant setting,
the expanded Lie generator is not consistent with the principles of perturbation analysis~\cite{Nayfeh} because the dependent variables are not expanded; consequently, in some cases, 
the approximately invariant solutions obtained by this method are not the most general ones (see~\cite{DSGO:lieapprox}). 
In contrast, the FS approach has a simple and coherent basis, but it requires a lot of computations (especially for higher order perturbations), and the crucial assumption of a fully coupled system is too strong, because the equations at a level should not be influenced by those at higher levels.
It is to be noticed that some variants~\cite{DolapciPakdemirli2004,Valenti} of the FS method have been proposed with the aim of reducing the amount of computations.

In a recent paper~\cite{DSGO:lieapprox}, an approximate symmetry theory, which is consistent with perturbative analysis and inherits the relevant properties of exact Lie symmetries of differential equations, has been proposed. More precisely, the dependent 
variables are expanded in power series of the small parameter as done in classical perturbative analysis; then, the Lie generator, assumed dependent on the small parameter, is accordingly expanded, and the approximate invariance  with respect to the approximate Lie generator is imposed, as in BGI method. Some applications of the new approach to approximate Lie symmetries of differential equations~\cite{DSGO:lieapprox} can be found in~\cite{Gorgone-2018,GorgoneOliveriEDJE-2018,GorgoneOliveriZAMP-2021}. Needless to say that
the method may require a lot of cumbersome computations; remarkably, there is a general and freely available package (ReLie,~\cite{Oliveri-relie}), written in the computer algebra system Reduce~\cite{Reduce}, able to do almost 
automatically all the needed work.

In this paper, we apply this consistent approach to variational problems, and state a consistent approximate Noether theorem leading to the construction of approximate conservation laws for models admitting a Lagrangian function containing small terms.

The plan of the paper is as follows. In Section~\ref{sec:theory}, to keep the paper self-contained and fix the notation, a brief sketch of the new approach to approximate Lie symmetries of differential equations introduced in~\cite{DSGO:lieapprox} is given. In Section~\ref{sec:Noether}, we describe the 
procedure for the approximate variational Lie symmetries of a Lagrangian function involving small perturbative terms, and state the approximate Noether theorem.
Section~\ref{sec:applications} contains some illustrative applications, and Section~\ref{sec:conclusions} our conclusions.

\section{The Consistent Approach to Approximate Lie Symmetries}
\label{sec:theory}

In this Section, in order to fix the notation, we briefly review the approach to approximate 
Lie symmetries of differential equations proposed in~\cite{DSGO:lieapprox}.
Let
\begin{equation}
\label{system}
\Delta\left(\mathbf{x},\mathbf{u},\mathbf{u}^{(r)};\varepsilon\right)=0
\end{equation}
be a differential equation of order $r$ where some terms include a parameter $\varepsilon\ll 1$;
function $\Delta$, assumed to be sufficiently smooth, depends on  the independent 
variables  $\mathbf{x}\in X\subseteq\mathbb{R}^n$, the dependent ones 
$\mathbf{u}\in U\subseteq\mathbb{R}^m$, and the derivatives (denoted by $\mathbf{u}^{(r)}$)  of the latter with respect to the former up to the order $r$. 
It is not uncommon to observe that Equation (\ref{system}) possesses few symmetries compared with the unperturbed equation  $\Delta\left(\mathbf{x},\mathbf{u},\mathbf{u}^{(r)};0\right)=0$.
Nevertheless, differential equations like (\ref{system}) are often investigated by means of 
a perturbative approach~\cite{Nayfeh} where one looks for
solutions in the form
\begin{equation}
\label{expansion_u}
\mathbf{u}(\mathbf{x};\varepsilon)=\sum_{k=0}^p\varepsilon^k \mathbf{u}_{(k)}(\mathbf{x})
+O(\varepsilon^{p+1}).
\end{equation}
Thus, Equation (\ref{system}) becomes 
\begin{equation}
\Delta\approx \sum_{k=0}^p\varepsilon^k\widetilde{\Delta}_{(k)}\left(\mathbf{x},\mathbf{u}_{(0)},
\mathbf{u}^{(r)}_{(0)},
\ldots,\mathbf{u}_{(k)},\mathbf{u}^{(r)}_{(k)}\right)=0,
\end{equation}
where, for any couple of  functions $f$ and $g$, the notation $f\approx g$ stands 
for $f-g=O(\varepsilon^{p+1})$.

The basic assumption of the approach to the approximate symmetries we want to use consists
in taking a Lie generator
\begin{equation}\label{gen-approx}
\Xi=\sum_{i=1}^n\xi_i(\mathbf{x},\mathbf{u};\varepsilon)\frac{\partial}{\partial x_i}
+\sum_{\alpha=1}^m\eta_\alpha(\mathbf{x},\mathbf{u};\varepsilon)\frac{\partial}{\partial u_\alpha}
\end{equation}
such that the infinitesimals  $\xi_i$ and $\eta_\alpha$ depend on the small parameter $\varepsilon$.

Inserting expansion~\eqref{expansion_u}, the infinitesimals can be written as
\begin{equation}\label{xi-approx1}
\xi_i\approx\sum_{k=0}^p\varepsilon^k \widetilde{\xi}_{(k)i}, \qquad \eta_\alpha\approx\sum_{k=0}
^p\varepsilon^k\widetilde{\eta}_{(k)\alpha},
\end{equation}
with
\begin{equation}\label{xi-approx2}
\begin{aligned}
&\widetilde{\xi}_{(0)i}=\xi_{(0)i}=\xi_i(\mathbf{x},\mathbf{u}_{(0)};0),\qquad
&&\widetilde{\eta}_{(0)\alpha}=\eta_{(0)\alpha}=\eta_\alpha(\mathbf{x},\mathbf{u}_{(0)};0),\\
&\widetilde{\xi}_{(k+1)i}=\frac{1}{k+1}\mathcal{R}[\widetilde{\xi}_{(k)i}],\qquad &&\widetilde{\eta}_{(k
+1)\alpha}=\frac{1}
{k+1}\mathcal{R}[\widetilde{\eta}_{(k)\alpha}],
\end{aligned}
\end{equation}
where $\mathcal{R}$ is a recursion operator, satisfying the Leibniz rule,  defined as 
\begin{equation}
\label{R_operator_new}
\begin{aligned}
&\mathcal{R}\left[\frac{\partial^{|\tau|}{f}_{(k)}(\mathbf{x},\mathbf{u}_{(0)})}{\partial u_{(0)1}
^{\tau_1}\dots\partial 
u_{(0)m}^{\tau_m}}\right]=\frac{\partial^{|\tau|}{f}_{(k+1)}(\mathbf{x},\mathbf{u}_{(0)})}{\partial 
u_{(0)1}
^{\tau_1}\dots\partial u_{(0)m}^{\tau_m}}\\
&\phantom{\mathcal{R}\left[\frac{\partial^{|\tau|}{f}_{(k)}(\mathbf{x},\mathbf{u}_{(0)})}{\partial 
u_{(0)1}
^{\tau_1}\dots\partial u_{(0)m}^{\tau_m}}\right]}
+\sum_{i=1}^m\frac{\partial}{\partial u_{(0)i}}\left(\frac{\partial^{|\tau|} {f}_{(k)}(\mathbf{x},
\mathbf{u}_{(0)})}{\partial 
u_{(0)1}^{\tau_1}\dots\partial u_{(0)m}^{\tau_m}}\right)u_{(1)i},\\
&\mathcal{R}[u_{(k)j}]=(k+1)u_{(k+1)j},
\end{aligned}
\end{equation}
for $k\ge 0$,  $j=1,\ldots,m$, $|\tau|=\tau_1+\cdots+\tau_m$.

Therefore, we get an approximate Lie generator
\begin{equation}
\Xi\approx \sum_{k=0}^p\varepsilon^k\widetilde{\Xi}_{(k)},
\end{equation}
where
\begin{equation}
\widetilde{\Xi}_{(k)}=\sum_{i=1}^n\widetilde{\xi}_{(k)i}(\mathbf{x},\mathbf{u}_{(0)},\ldots,\mathbf{u}
_{(k)})
\frac{\partial}{\partial x_i}
+\sum_{\alpha=1}^m\widetilde{\eta}_{(k)\alpha}(\mathbf{x},\mathbf{u}_{(0)},\ldots,\mathbf{u}_{(k)})
\frac{\partial}{\partial u_\alpha}.
\end{equation}

According to Lie's algorithm, the approximate Lie generator must be prolonged to the order $r$ 
to account for the transformation of derivatives; the prolongation is algorithmically done requiring that the contact conditions are preserved~\cite{Ovsiannikov1982book,Olver1986book,BlumanAnco2002book,BlumanCheviakovAnco2009book}:
\begin{equation}\label{prolongations}
\begin{aligned}
&\Xi^{(0)}=\Xi,\\
&\Xi^{(r)}=\Xi^{(r-1)}+\sum_{\alpha=1}^m\sum_{i_1=1}^n\ldots\sum_{i_r=1}^n\eta_{\alpha,i_1\ldots i_r}
\frac{\partial}
{\partial \frac{\partial^r u_\alpha}{\partial x_{i_1}\ldots\partial x_{i_r}}},\qquad r>0,
\end{aligned}
\end{equation}
where
\begin{equation}
\label{extended-generators}
\eta_{\alpha,i_1\ldots i_r}=\frac{D \eta_{\alpha,i_1\ldots i_{r-1}}}{D x_{i_r}}-\sum_{k=1}^n\frac{D 
\xi_k}{D x_{i_r}}
\frac{\partial^r u_\alpha}{\partial x_{i_1}\ldots \partial x_{i_{r-1}}\partial x_k},
\end{equation}
along with the {approximate} Lie derivative defined as
\begin{equation}
\label{Lie-derivative}
D_i=\frac{{D}}{{D} x_i}=\frac{\partial}{\partial x_i}+\sum_{k=0}^p\sum_{\alpha=1}^m \left(u_{(k)\alpha,i}
\frac{\partial}{\partial u_{(k)\alpha}}+\sum_{j=1}^n u_{(k)\alpha,ij} \frac{\partial}{\partial u_{(k)\alpha,j}}+\dots\right),
\end{equation}
where $\displaystyle u_{(k)\alpha,i}=\frac{\partial u_{(k)\alpha}}{\partial x_{i}}$, 
$\displaystyle u_{(k)\alpha,ij}=\frac{\partial^2 u_{(k)\alpha}}{\partial x_{i}\partial x_j}$, \ldots

Of course, in the expression of prolongations (\ref{extended-generators}), 
we need to take into account the expansions of $\xi_i$, $\eta_\alpha$, $u_\alpha$, and drop the $O(\varepsilon^{p+1})$ terms, so that
\begin{equation}
\Xi^{(r)} \approx \sum_{k=0}^p \varepsilon^k\widetilde{\Xi}^{(r)}_{(k)},
\end{equation}
with obvious meaning of symbols.

The approximate (at the order $p$) invariance condition of a differential equation reads
\begin{equation}
\left.\sum_{k=0}^p \varepsilon^k \sum_{\ell=0}^k \widetilde{\Xi}^{(r)}_{(\ell)}\widetilde{\Delta}_{(k-\ell)}\right|_{\Delta\approx 0}\approx 0.
\end{equation}
Some immediate consequences of this definition are in order~\cite{DSGO:lieapprox}. 
The Lie generator $\widetilde{\Xi}_{(0)}$ is always a symmetry of the unperturbed equations ($
\varepsilon=0$), and  $\displaystyle\sum_{k=1}^p\varepsilon^k\widetilde{\Xi}_{(k)}$ 
provides the deformation of the symmetry 
due to the small contributions; in contrast, not all the symmetries of the unperturbed equations are admitted as the zeroth terms of the approximate Lie symmetries. Moreover, 
if $\Xi$ is the generator of an approximate Lie symmetry, then
$\varepsilon\Xi$ is an admitted generator too. Finally, 
the approximate Lie symmetries of a differential equation are the elements of an
approximate Lie algebra.

\section{Approximate Noether Theorem}
\label{sec:Noether}
In this Section, we will be concerned with Noether theorem in the framework of approximate Lie symmetries outlined in Section~\ref{sec:theory}.
\begin{Definition}
Given a system of differential equations,
\begin{equation}
\label{sysgen}
\boldsymbol{\Delta}\left(\mathbf{x},\mathbf{u},\mathbf{u}^{(r)};\varepsilon\right)\approx \sum_{k=0}^p\varepsilon^k\widetilde{\boldsymbol\Delta}_{(k)}\left(\mathbf{x},\mathbf{u}_{(0)},
\mathbf{u}^{(r)}_{(0)},
\ldots,\mathbf{u}_{(k)},\mathbf{u}^{(r)}_{(k)}\right)=\mathbf{0},
\end{equation}
an approximate conservation law of order $r$ compatible with the system \eqref{sysgen} is a divergence expression
\begin{equation}
\label{conservationlaw} 
 \sum_{k=0}^p\varepsilon^k\left(\sum_{i=1}^n D_{i}\left(\widetilde{\Phi}^i_{(k)}\left(\mathbf{x},\mathbf{u}_{(0)},
\mathbf{u}^{(r-1)}_{(0)},
\ldots,\mathbf{u}_{(k)},\mathbf{u}^{(r-1)}_{(k)}\right)\right)\right) =O(\varepsilon^{p+1}),
\end{equation}
holding for all solutions of system \eqref{sysgen}, where 
\[
\sum_{k=0}^p \varepsilon^k  \widetilde{\Phi}^i_{(k)}\left(\mathbf{x},\mathbf{u}_{(0)},
\mathbf{u}^{(r-1)}_{(0)},
\ldots,\mathbf{u}_{(k)},\mathbf{u}^{(r-1)}_{(k)}\right), \qquad i=1,\ldots,n
\]
are the expansions at order $p$ of the fluxes of the conservation law according to (\ref{expansion_u}).
\end{Definition}
An approximate conservation law is {trivial} if it is $O(\varepsilon^{p+1})$ identically.
The notion of trivial conservation laws allows us for the introduction of equivalent conservation laws. Two approximate conservation laws are said {equivalent} if their linear combination is a trivial approximate conservation law. In general, a finite set $\mathcal{S}$ of approximate conservation laws is linearly dependent if there exists a set of constants, not all zero, such that the linear combination of the elements in $\mathcal{S}$ is trivial. In this case, at least one of the approximate conservation laws in $\mathcal{S}$ can be expressed as a linear combination of the remaining ones. 

It is well known that for variational problems, {i.e.}, differential equations derived from a Lagrangian function, the determination of conservation laws is ruled by Noether's theorem~\cite{Noether},  establishing a correspondence between symmetries of the action integral and conservation laws 
through an explicit formula involving the infinitesimal generator of the Lie  symmetry 
and the Lagrangian itself. 

In the sequel, let us limit ourselves to Lagrangian functions involving first order derivatives, albeit the procedure can be straightforwardly extended in the case of  higher order Lagrangians.
 
Let 
\begin{equation}\label{lag-function}
\mathcal{L}\left(\mathbf{x},\mathbf{u},\mathbf{u}^{(1)};\varepsilon\right)
\end{equation}
a Lagrangian function involving a small parameter $\varepsilon$. After inserting the expansion \eqref{expansion_u} in~\mbox{(\ref{lag-function})}, we have
\begin{equation}
\label{lag-expansion}
\begin{aligned}
\mathcal{L}\left(\mathbf{x},\mathbf{u},\mathbf{u}^{(1)};\varepsilon\right)&
\approx \mathcal{L}_0\left(\mathbf{x},\mathbf{u}_{(0)},\mathbf{u}_{(0)}^{(1)}\right)\\
&+
\sum_{k=1}^p \varepsilon^k \mathcal{L}_k\left(\mathbf{x},\mathbf{u}_{(0)},\ldots,\mathbf{u}_{(k)},\mathbf{u}_{(0)}^{(1)},\ldots,\mathbf{u}_{(k)}^{(1)}\right),
\end{aligned}
\end{equation}
where $\mathcal{L}_0=\mathcal{L}\left(\mathbf{x},\mathbf{u}_{(0)},\mathbf{u}_{(0)}^{(1)};0\right)$ is the unperturbed Lagrangian and the remaining terms in~\mbox{(\ref{lag-expansion})} represent the perturbation up to the order $p$ in $\varepsilon$.

Let us consider the Lagrangian action
\begin{equation}
\label{action}
\begin{aligned}
\mathcal{J}\left(\mathbf{x},\mathbf{u},\mathbf{u}^{(1)};\varepsilon\right)&= \int_\Omega
\mathcal{L}\left(\mathbf{x},\mathbf{u},\mathbf{u}^{(1)};\varepsilon\right)d\mathbf{x} \approx\\
&\approx\int_\Omega \left(
\sum_{k=0}^p\varepsilon^k\mathcal{L}_k\left(\mathbf{x},\mathbf{u}_{(0)},\ldots,\mathbf{u}_{(k)},\mathbf{u}^{(1)}_{(0)},\ldots,\mathbf{u}^{(1)}_{(k)}\right)\right)d\mathbf{x},
\end{aligned}
\end{equation}
over a domain $\Omega$. By requiring the 
first variation of the Lagrangian action (\ref{action}) to be $O(\varepsilon^{p+1})$ under variations of order $O(\varepsilon^{p+1})$ at the boundary of $\Omega$, we obtain the approximate Lagrange equations
\begin{equation}
\sum_{k=0}^p \varepsilon^k \left(\frac{\partial \mathcal{L}_k}{\partial u_{(0)\alpha}}
-\sum_{i=1}^nD_{i}\left(\frac{\partial \mathcal{L}_k
}{\partial u_{(0)\alpha,i}}\right)\right)\approx 0, \qquad \alpha=1,\ldots,m.
\end{equation}

Now, we have all the ingredients to state the approximate Noether theorem.

\begin{Theorem}[Approximate Noether theorem]
\label{Noether-theorem}
Let us consider a variational system of differential equations arising from a first order Lagrangian function.
Let expression (\ref{gen-approx}) be the generator of an approximate Lie symmetry 
of the approximate Lagrangian action (\ref{action}), say
\begin{equation}
\label{inv-action}
\sum_{k=0}^{p}\varepsilon^k\left(\sum_{j=0}^k\left(\widetilde{\Xi}^{(1)}_{(j)}\mathcal{L}_{k-j}+\mathcal{L}_{k-j}\sum_{i=1}^n D_{i} \widetilde{\xi}_{(j)i}\right)-\sum_{i=1}^n D_{i}\phi^i_{(k)}\right) \approx 0,
\end{equation}
where  $\phi^i_{(k)}\left(\mathbf{x},\mathbf{u}_{(0)},\ldots,\mathbf{u}_{(k)}\right)$  ($i=1,\ldots,n$)
are functions of their arguments to be suitably determined.
Then, the approximate conservation law 
\begin{equation}
\sum_{k=0}^p\varepsilon^k \sum_{i=1}^n D_i \widetilde{\Phi}^i_{(k)}\approx 0,
\end{equation}
where
\begin{equation}
\widetilde{\Phi}^i_{(k)}=\sum_{\ell=0}^k\left(\sum_{\alpha=1}^m\left(\left(\widetilde{\eta}_{(\ell)\alpha}-\sum_{j=1}^n\widetilde{\xi}_{(\ell)j} u_{(\ell)\alpha,j}\right)\sum_{q=0}^{k-\ell}\frac{\partial\mathcal{L}_{k-\ell}}{\partial u_{(q)\alpha,i}}\right)+\widetilde{\xi}_{(\ell)i}\mathcal{L}_{k-\ell}\right)-\phi^i_{(k)},
\end{equation}
is recovered.

In the case of only one independent variable, say the time $t$, the approximate Noether theorem
provides a conserved quantity
\[
\mathcal{I}=\sum_{k=0}^p \varepsilon^k \widetilde{\Phi}_{(k)}
\]
such that
\[
\frac{D\mathcal{I}}{Dt}\approx 0.
\]
\end{Theorem}
\begin{proof}
The proof is immediate and is simply an adaptation of the proof of classical Noether theorem
in the approximate context.
 \end{proof}

\begin{Remark}
A Noether theorem, as well as an approximate Noether theorem, can be easily stated in the case of
higher order Lagrangian functions (see, for instance,~\cite{BlumanCheviakovAnco2009book} for the exact case).
\end{Remark}

\begin{Remark}
The generators of approximate variational Lie symmetries leave the Lagrangian action approximately invariant; consequently, they are admitted as approximate Lie symmetries of the approximate Lagrange equations. As in the exact case, the converse is not true, since not all the approximate Lie symmetries of the Euler--Lagrange equations are approximate variational Lie symmetries.
\end{Remark}

\section{Applications}
\label{sec:applications}
In this Section, we consider some examples of physical interest where the procedure for the approximate variational Lie symmetries can be applied. The first two examples have been considered in literature and analyzed by means of the BGI approach to approximate Lie symmetries, the third one faces the planar three body problem where one of the material points has a mass of order $\varepsilon$. 
Moreover,  we limit ourselves to first order approximate variational Lie symmetries of first order perturbed Lagrangian functions, and determine the corresponding  first order approximate conservation laws. The computation of approximate variational Lie symmetries and approximate conservation laws has been done by means of the program ReLie~\cite{Oliveri-relie}.

\subsection{Perturbed Harmonic Oscillator}
As first example, we consider the one-dimensional perturbed harmonic oscillator, say the differential equation
\begin{equation}\label{oscillar}
\ddot{u}+ u+\varepsilon F(u)=0,
\end{equation}
derived from the Lagrangian
\begin{equation}
\label{lagrangian-harmonic}
\mathcal{L}(u,\dot u;\varepsilon)=\frac{1}{2}\left(\dot u^2-u^2\right)-\varepsilon \int F(u)du, 
\end{equation}
where the dot stands for differentiation with respect to time, and $F(u)$ is an arbitrary function of $u$ such that $\displaystyle \frac{d^2F}{du^2}\neq 0$. 
This model has been studied in~\cite{Baikov-1994} for the classification of the admitted approximate Lie symmetries, and in~\cite{Govinder-1998}, where an analysis of the approximate Noether symmetries and the corresponding approximate conservation laws has been developed.
In order to perform the algorithmic procedure for approximate variational 
Lie symmetries described in the previous Section, let us
expand $u(t;\varepsilon)$ at ﬁrst order in $\varepsilon$, {i.e.},
\begin{equation}
u(t;\varepsilon)=u_{(0)}(t)+\varepsilon u_{(1)}(t)+ O(\varepsilon^2),
\end{equation}
whereupon we have the first order perturbed Lagrangian 
\begin{equation}
\begin{aligned}
\mathcal{L}&\equiv\mathcal{L}_0\left(u_{(0)},\dot u_{(0)}\right)+\varepsilon \mathcal{L}_1\left(u_{(0)},u_{(1)},\dot u_{(0)},\dot u_{(1)}\right)=\\
&=\frac{1}{2}\left(\dot u_{(0)}^{2}-u_{(0)}^2\right)+\varepsilon\left(\dot u_{(0)} \dot u_{(1)}-u_{(0)}u_{(1)}-\int F_0(u_{(0)}) du_{(0)}\right),
\end{aligned}
\end{equation}
with $\displaystyle F(u)\approx F_0(u_{(0)})+\varepsilon \frac{d F_0(u_{(0)})}{d u_{(0)}}u_{(1)}$.

By solving the approximate determining equations, a classification is needed  according to the functional form of $F(u)$:
\begin{itemize}
\item[(1)] $F(u)$ arbitrary;
\item[(2)] $F(u)=(u+\delta)^2$, $\delta$ constant;
\item[(3)] $\displaystyle F(u)=\frac{\kappa}{u^3}$, $\kappa$ constant.
\end{itemize}

For $F(u)$ arbitrary, we determine the following generators of approximate variational Lie symmetries together with the expansion of the associated function $\phi$ entering the invariance condition (\ref{inv-action}):
\begin{equation}
\label{sym_arb}
\begin{aligned}
&\Xi_1=\frac{\partial}{\partial t},\qquad &&\phi_{(0)}=\phi_{(1)}=0;\\
&\Xi_2=\varepsilon \sin(t)\frac{\partial}{\partial u},\qquad  &&\phi_{(0)}=0,\quad\phi_{(1)}=-\cos(t)u_{(0)};\\
&\Xi_3=\varepsilon \cos(t)\frac{\partial}{\partial u},\qquad &&\phi_{(0)}=0,\quad\phi_{(1)}=\sin(t)u_{(0)};\\
&\Xi_4=\varepsilon\left(\sin(2t)\frac{\partial}{\partial t}+\cos(2t)u_{(0)}\frac{\partial}{\partial u}\right),
\qquad&&\phi_{(0)}=0,\quad\phi_{(1)}=\sin(2t)u_{(0)}^2;\\
&\Xi_5=\varepsilon\left(\cos(2t)\frac{\partial}{\partial t}-\sin(2t)u_{(0)}\frac{\partial}{\partial u}\right),
\qquad &&\phi_{(0)}=0,\quad \phi_{(1)}=\cos(2t)u_{(0)}^2;\\
&\Xi_6=\varepsilon\Xi_1,\qquad &&\phi_{(0)}=\phi_{(1)}=0;
\end{aligned}
\end{equation} 
the application of Theorem \ref{Noether-theorem} yields the following approximate conserved quantities:
\[
\begin{aligned}
&\mathcal{I}_1=\frac{1}{2}\left(\dot u_{(0)}^{2}+u_{(0)}^2\right)+\varepsilon\left(\dot u_{(0)} \dot u_{(1)}+u_{(0)}u_{(1)}+\int F_0(u_{(0)}) du_{(0)}\right),\\
&\mathcal{I}_2=\varepsilon \left(\sin(t)\dot u_{(0)}-\cos(t)u_{(0)}\right),\\
&\mathcal{I}_3=\varepsilon \left(\cos(t)\dot u_{(0)}+\sin(t)u_{(0)}\right),\\
&\mathcal{I}_4=\varepsilon \left(\left(\sin(t)\dot u_{(0)}-\cos(t)u_{(0)}\right)\left(\cos(t)\dot u_{(0)}+\sin(t)u_{(0)}\right)\right),\\
&\mathcal{I}_5=\varepsilon \left(\sin(t)\dot u_{(0)}-\cos(t)u_{(0)}\right)^2,\\
&\mathcal{I}_6=\varepsilon \mathcal{I}_1.
\end{aligned}
\]

For $F(u)=(u+\delta)^2$, in addition to \eqref{sym_arb}, we have two further generators of approximate variational Lie symmetries, namely 
\[
\begin{aligned}
&\Xi_{7a}=4\varepsilon\sin(t)\frac{\partial}{\partial t}+\left(-3\cos(t)+\varepsilon\left(3\delta t\sin(t)+2\cos(t)u_{(0)}\right)\right)\frac{\partial}{\partial u},\\
&\qquad\phi_{(0)}=-3\sin(t)u_{(0)},\\
&\qquad\phi_{(1)}=\sin(t)u_{(0)}^2-3\delta(t\cos(t)+\sin(t))u_{(0)}-3\sin(t)u_{(1)}-3\delta^2\sin(t);\\
&\Xi_{8a}=4\varepsilon\cos(t)\frac{\partial}{\partial t}+\left(3\sin(t)+\varepsilon\left(3\delta t\cos(t)-2\sin(t)u_{(0)}\right)\right)\frac{\partial}{\partial u},\\
&\qquad\phi_{(0)}=-3\cos(t)u_{(0)},\\
&\qquad\phi_{(1)}=\cos(t)u_{(0)}^2+3\delta(t\sin(t)-\cos(t))u_{(0)}-3\cos(t)u_{(1)}-3\delta^2\cos(t);
\end{aligned}
\]
the associated approximate conserved quantities turn out to be
\[
\begin{aligned}
\mathcal{I}_{7a}&=\cos(t)\dot u_{(0)}+\sin(t)u_{(0)}+\varepsilon\left(\frac{2}{3}\sin(t)\dot u_{(0)}^{2}-\left(\frac{2}{3}\cos(t)u_{(0)}+\delta t\sin(t)\right)\dot u_{(0)}\right.\\
&+\left.\frac{\sin(t)}{3}u_{(0)}^2+\delta(t\cos(t)+\sin(t))u_{(0)}+\cos(t)\dot u_{(1)}+\sin(t)u_{(1)}+\delta^2\sin(t)\right),\\
\mathcal{I}_{8a}&=\sin(t)\dot u_{(0)}-\cos(t)u_{(0)}+\varepsilon\left(-\frac{2}{3}\cos(t)\dot u_{(0)}^{2}-\left(\frac{2}{3}\sin(t)u_{(0)}-\delta t\cos(t)\right)\dot u_{(0)}\right.\\
&-\left.\frac{\cos(t)}{3}u_{(0)}^2+\delta(t\sin(t)-\cos(t))u_{(0)}+\sin(t)\dot u_{(1)}-\cos(t)u_{(1)}-\delta^2\cos(t)\right).
\end{aligned}
\]

Finally, for $\displaystyle F(u)=\frac{\kappa}{u^3}$, we have, besides generators (\ref{sym_arb}), 
the following generators of approximate variational Lie symmetries, say
\[
\begin{aligned}
&\Xi_{7b}=\cos(2t)\frac{\partial}{\partial t}-\sin(2t)\left(u_{(0)}+\varepsilon u_{(1)}\right)\frac{\partial}{\partial u}\\
&\qquad\phi_{(0)}=\cos(2t)u_{(0)}^2,\quad\phi_{(1)}=2\cos(2t)u_{(0)}u_{(1)};\\
&\Xi_{8b}=\sin(2t)\frac{\partial}{\partial t}+\cos(2t)\left(u_{(0)}+\varepsilon u_{(1)}\right)\frac{\partial}{\partial u},\\
&\qquad\phi_{(0)}=\sin(2t)u_{(0)}^2,\quad\phi_{(1)}=2\sin(2t)u_{(0)}u_{(1)};
\end{aligned}
\]
thus, the following approximate conserved quantities are obtained:
\[
\begin{aligned}
\mathcal{I}_{7b}&=\cos(2t)\frac{\dot u_{(0)}^{2}-u_{(0)}^{2}}{2}+\sin(2t)\dot u_{(0)} u_{(0)}+\varepsilon\left(\left(\cos(2t)\dot u_{(1)}+\sin(2t)u_{(1)}\right)\dot u_{(0)}\phantom{\frac{\cos(2t)}{2 u_{(0)}^{2}}}\right.\\
&+\left.\left(\sin(2t)\dot u_{(1)}-\cos(2t)u_{(1)}\right)u_{(0)}-\kappa\frac{\cos(2t)}{2 u_{(0)}^{2}}\right),\\
\mathcal{I}_{8b}&=\sin(2t)\frac{\dot u_{(0)}^{2}-u_{(0)}^2}{2}-\cos(2t)\dot u_{(0)} u_{(0)}+\varepsilon\left(\left(\sin(2t)\dot u_{(1)}-\cos(2t)u_{(1)}\right)\dot u_{(0)}\phantom{\frac{\cos(2t)}{2 u_{(0)}^{2}}}\right.\\
&-\left.\left(\cos(2t)\dot u_{(1)}+\sin(2t)u_{(1)}\right)u_{(0)}-\kappa\frac{\sin(2t)}{2 u_{(0)}^{2}}\right).
\end{aligned}
\]
\begin{Remark}
The approximate Noether symmetries above recovered are just the expansion of those determined in~\cite{Govinder-1998}; this explains essentially the concept that our approach, compared with the BGI method,  is consistent with the principles of perturbation analysis; moreover, the corresponding first order approximate conserved quantities obtained in this paper are somehow different from those reported in~\cite{Govinder-1998} (maybe some formulas of the latter paper are not completely correct: We underline that, albeit in~\cite{Govinder-1998} first order approximate conserved quantities are considered, the results there provided include also higher order terms in $\varepsilon$).
\end{Remark}

\subsection{A Second Order System of ODEs}
Let us consider the following perturbed second order system of ordinary differential equations investigated in~\cite{Campoamor-Stursberg-2016},
\begin{equation}
\label{system-ode}
\begin{aligned}
&\ddot u+\frac{\alpha}{u^3}-\varepsilon\frac{1}{u^3}\frac{dF(v)}{dv}=0,\\
&\ddot v-4\alpha\frac{v}{u^4}+2\frac{1}{u}\dot u \dot v+2\varepsilon\frac{1}{u^4}\frac{d}{dv}\left(vF(v)\right)=0;
\end{aligned}
\end{equation}
system (\ref{system-ode}) can be derived from the Lagrangian
\begin{equation}
\mathcal{L}(u,v,\dot u,\dot v;\varepsilon)=v \dot u^{2}+ u \dot u \dot v -\alpha\frac{v}{u^2}+\varepsilon\frac{F(v)}{u^2},
\end{equation}
where $u\equiv u(t;\varepsilon)$, $v\equiv v(t;\varepsilon)$, $F(v)$ arbitrary smooth function of its argument, and\linebreak $\alpha$~constant.

In~\cite{Campoamor-Stursberg-2016},  system (\ref{system-ode}) was analyzed as an approximate variational problem, and some approximate conserved quantities for arbitrary $F(v)$ have been determined. As our consistent procedure requires, expanding  $u(t;\varepsilon)$ and $v(t;\varepsilon)$ at first order in $\varepsilon$, we get the first order perturbed Lagrangian

\begin{equation}
\begin{aligned}
\mathcal{L}&\approx \mathcal{L}_0\left(u_{(0)},v_{(0)},\dot u_{(0)},\dot v_{(0)}\right)+
\varepsilon\mathcal{L}_1\left(u_{(0)},v_{(0)},u_{(1)},v_{(1)},
\dot u_{(0)},\dot v_{(0)},\dot u_{(1)},\dot v_{(1)}\right) =\\
&=v_{(0)}\dot u_{(0)}^2 +u_{(0)}\dot u_{(0)}\dot v_{(0)}-\alpha \frac{v_{(0)}}{u_{(0)}^2}\\
&+\varepsilon\left(2v_{(0)}\dot u_{(0)}\dot u_{(1)}+v_{(1)}\dot u_{(0)}^2+u_{(0)}\dot u_{(0)}\dot v_{(1)}+u_{(0)}\dot u_{(1)}\dot v_{(0)}+u_{(1)}\dot u_{(0)}\dot v_{(0)}\phantom{\frac{u_{(0)}}{v_{(0)}}}\right.\\
&\left.-\alpha\left(\frac{v_{(1)}}{u_{(0)}^2}-
2\frac{v_{(0)}}{u_{(0)}^3}u_{(1)}\right)+\frac{F(v_{(0)})}{u_{(0)}^2}\right);
\end{aligned}
\end{equation}
by imposing the approximate invariance of the action integral, we are able to obtain the following approximate generators together with the expansion of the associated function $\phi$ entering  condition (\ref{inv-action}):
\begin{equation}\label{sym-ode}
\begin{aligned}
&\Xi_1=\frac{\partial}{\partial t},\qquad && \phi_{(0)}=\phi_{(1)}=0;\\
&\Xi_2=t^2\frac{\partial}{\partial t}+t\left(u_{(0)}+\varepsilon u_{(1)}\right)\frac{\partial}{\partial u},\qquad&& \phi_{(0)}=-u_{(0)}^2 v_{(0)},\\
& \qquad &&\phi_{(1)}=-u_{(0)}\left(u_{(0)}v_{(1)}+2v_{(0)}u_{(1)}\right);\\
&\Xi_3=2t\frac{\partial}{\partial t}+\left(u_{(0)}+\varepsilon u_{(1)}\right)\frac{\partial}{\partial u},&& \phi_{(0)}=\phi_{(1)}=0;\\
&\Xi_4=\varepsilon\Xi_1,\qquad && \phi_{(0)}=\phi_{(1)}=0;\\
&\Xi_5=\varepsilon\Xi_2,\qquad &&\phi_{(0)}=0,\quad \phi_{(1)}=-u_{(0)}^2 v_{(0)};\\
&\Xi_6=\varepsilon\Xi_3,\qquad && \phi_{(0)}=\phi_{(1)}=0;
\end{aligned}
\end{equation}
these generators provide the following approximate conserved quantities:
\begin{align*}
\mathcal{I}_{1}&=\dot u_{(0)}^{2}v_{(0)}+\dot u_{(0)}\dot v_{(0)}u_{(0)}+\alpha \frac{v_{(0)}}{u_{(0)}^2}+
\varepsilon\left(\dot u_{(0)}^{2}v_{(1)}+\dot u_{(0)}\dot v_{(0)}u_{(1)}+2 \dot u_{(0)}\dot u_{(1)} v_{(0)}\phantom{\frac{F_0}{u_{(0)}^2}}\right.\\
&\left.+\dot u_{(0)}\dot v_{(1)} u_{(0)}+\dot v_{(0)}\dot u_{(1)} u_{(0)}+\alpha\frac{u_{(0)}v_{(1)}-2v_{(0)}u_{(1)}}{u_{(0)}^3}-\frac{F_0}{u_{(0)}^2}\right),\allowdisplaybreaks\\
\mathcal{I}_{2}&=t^2\left(\dot u_{(0)}^{2}v_{(0)}+\dot u_{(0)}\dot v_{(0)} u_{(0)}+\alpha \frac{v_{(0)}}{u_{(0)}^2}\right)-t\left(2 \dot u_{(0)} u_{(0)} v_{(0)}+ \dot v_{(0)} u_{(0)}^2\right)+u_{(0)}^2 v_{(0)}\\
&+\varepsilon\left(t^2\left(\dot u_{(0)}^{2}v_{(1)}+\dot u_{(0)}\dot v_{(0)} u_{(1)}+2 \dot u_{(0)}\dot u_{(1)} v_{(0)}+\dot u_{(0)} \dot v_{(1)} u_{(0)}+\dot v_{(0)} \dot u_{(1)}u_{(0)}\phantom{\frac{F_0}{u_{(0)}^2}}\right.\right.\\
&+\left.\left.\alpha\frac{u_{(0)}v_{(1)}-2v_{(0)}u_{(1)}}{u_{(0)}^3}-\frac{F_0}{u_{(0)}^2}\right)-t\left(2 \dot u_{(0)}\left(u_{(0)}v_{(1)}+u_{(1)}v_{(0)}\right)+2 \dot v_{(0)} u_{(0)}u_{(1)}\right.\right.\\
&\left.\left.+2 \dot u_{(1)} u_{(0)}v_{(0)}+\dot v_{(1)} u_{(0)}^2\right)+u_{(0)}^2v_{(1)}+2u_{(0)}v_{(0)}u_{(1)}\right),\allowdisplaybreaks\\
\mathcal{I}_{3}&=2t\left(\dot u_{(0)}^{2}v_{(0)}+\dot u_{(0)}\dot v_{(0)}u_{(0)}+\alpha \frac{v_{(0)}}{u_{(0)}^2}\right)-2\dot u_{(0)} u_{(0)} v_{(0)}-\dot v_{(0)} u_{(0)}^2\\
&+\varepsilon\left(2t\left(\dot u_{(0)}^{2}v_{(1)}+ \dot u_{(0)}\dot v_{(0)} u_{(1)}+2 \dot u_{(0)}\dot u_{(1)} v_{(0)}+\dot u_{(0)}\dot v_{(1)}u_{(0)}+\dot v_{(0)}\dot u_{(1)}u_{(0)}\phantom{\frac{F_0}{u_{(0)}^2}}\right.\right.\\
&\left.\left.+\alpha\frac{u_{(0)}v_{(1)}-2v_{(0)}u_{(1)}}{u_{(0)}^3}-\frac{F_0}{u_{(0)}^2}\right)-\left(2 \dot u_{(0)}\left(u_{(0)}v_{(1)}+u_{(1)}v_{(0)}\right)+2\dot v_{(0)} u_{(0)}u_{(1)}\right.\right.\\
&\left.\left.+\dot 2u_{(1)} u_{(0)}v_{(0)}+\dot v_{(1)} u_{(0)}^2\right)\right),\allowdisplaybreaks\\
\mathcal{I}_{4}&=\varepsilon\mathcal{I}_{1},\qquad\mathcal{I}_{5}=\varepsilon\mathcal{I}_{2},\qquad\mathcal{I}_{6}=\varepsilon\mathcal{I}_{3},
\end{align*}
where $\displaystyle F(v)\approx F_0(v_{(0)})+\varepsilon \frac{d F_0(v_{(0)})}{d v_{(0)}}v_{(1)}$.
\begin{Remark}
In this application, we recovered a set of approximate conserved quantities different from that reported in~\cite{Campoamor-Stursberg-2016}; furthermore, there is no need to consider special instances of the arbitrary function $F(v)$.
\end{Remark}

\subsection{The Three-Body Problem}
Let us now focus on the problem of three bodies with masses $m_{\alpha}$ ($\alpha=1,2,3$) moving under their mutual gravitational attraction. We assume that the mass $m_3$ of the third body is much smaller than $m_1$ and $m_2$, and that the three point masses move  
in a fixed plane. Therefore, we consider a planar restricted problem without neglecting the gravitational action of the third body on the two main bodies. Let $\mathbf{r}_i\equiv(x_i,y_i,0)$ $(i=1,2,3)$ be the position vectors of the three point masses in a fixed frame reference, and $\mathbf{r}_{ij}=\mathbf{r}_i-\mathbf{r_j}$ $(1\le i<j\le 3)$. Under these hypotheses, the motion equations are:
\begin{equation}
\label{3body-motion}
\begin{aligned}
&\ddot{r}_1+Gm_2 \frac{\mathbf{r}_{12}}{|\mathbf{r}_{12}|^3}+\varepsilon Gm_3 \frac{\mathbf{r}_{13}}{|\mathbf{r}_{13}|^3}=\mathbf{0},\\
&\ddot{r}_2-Gm_1 \frac{\mathbf{r}_{12}}{|\mathbf{r}_{12}|^3}+\varepsilon Gm_3 \frac{\mathbf{r}_{23}}{|\mathbf{r}_{23}|^3}=\mathbf{0},\\
&\ddot{r}_3-Gm_1 \frac{\mathbf{r}_{13}}{|\mathbf{r}_{13}|^3}- Gm_2 \frac{\mathbf{r}_{23}}{|\mathbf{r}_{23}|^3}=\mathbf{0},
\end{aligned}
\end{equation}
where $G$ is the gravitational constant. Equation~(\ref{3body-motion}) derive from the Lagrangian function 
\begin{equation}
\mathcal{L}=\frac{1}{2}\left(m_1\dot{\mathbf{r}}_1^2+m_2\dot{\mathbf{r}}_2^2\right)
+\frac{G m_1 m_2}{|\mathbf{r}_{12}|}
+\varepsilon\left(m_3\dot{\mathbf{r}}_3^2+\frac{G m_1 m_3}{|\mathbf{r}_{13}|}+\frac{G m_2 m_3}{|\mathbf{r}_{23}|}\right).
\end{equation}

Let us expand the dependent variables at ﬁrst order in $\varepsilon$, {i.e.},
\begin{equation}
\begin{aligned}
\mathbf{r}_i&=\mathbf{r}_{(0)i}+\varepsilon\mathbf{r}_{(1)i}+O(\varepsilon^2)\equiv\\
&\equiv\left(x_{(0)i}(t)+\varepsilon x_{(1)i}(t)+ O(\varepsilon^2),y_{(0)i}(t)+\varepsilon y_{(1)i}(t)+ O(\varepsilon^2),0\right), \qquad i=1,2,3,
\end{aligned}
\end{equation}
whereupon we have
\begin{equation}
\mathbf{r}_{ij}=\mathbf{r}_{(0)ij}+\varepsilon\mathbf{r}_{(1)ij}+O(\varepsilon^2)=
\mathbf{r}_{(0)i}-\mathbf{r}_{(0)j}+\varepsilon\left(\mathbf{r}_{(1)i}-\mathbf{r}_{(1)j}\right)+O(\varepsilon^2),
\end{equation}
and
\begin{equation}
\begin{aligned}
\mathcal{L}&\approx \mathcal{L}_0+\varepsilon\mathcal{L}_1=\\
&=\frac{1}{2}\left(m_1\dot{\mathbf{r}}_{(0)1}^2+m_2\dot{\mathbf{r}}_{(0)2}^2\right)
+\frac{G m_1 m_2}{|\mathbf{r}_{(0)12}|}\\
&+\varepsilon\left(m_1\dot{\mathbf{r}}_{(0)1}\dot{\mathbf{r}}_{(1)1}+m_2\dot{\mathbf{r}}_{(0)2}\cdot\dot{\mathbf{r}}_{(1)2} +m_3\dot{\mathbf{r}}_{(0)3}^2\phantom{\frac{G m_2 m_3}{|\mathbf{r}_{(0)23}|}}\right.\\
&\quad\left.-\frac{Gm_1m_2}{|\mathbf{r}_{(0)12}|^3}\mathbf{r}_{(0)12}\cdot\mathbf{r}_{(1)12} +\frac{G m_1 m_3}{|\mathbf{r}_{(0)13}|}+\frac{G m_2 m_3}{|\mathbf{r}_{(0)23}|}\right);
\end{aligned}
\end{equation}
by searching for the approximate variational Lie symmetries, we are able to determine the following approximate generators together with the expansion of the associated function $\phi$ entering the invariance condition (\ref{inv-action}):
\begin{align*}
&\Xi_1=\frac{\partial}{\partial t},\qquad \phi_{(0)}=\phi_{(1)}=0;\allowdisplaybreaks\\
&\Xi_{2a}=\sum_{i=1}^3\frac{\partial}{\partial x_i},\qquad\Xi_{2b}=\sum_{i=1}^3\frac{\partial}{\partial y_i},\qquad  \phi_{(0)}=\phi_{(1)}=0;\allowdisplaybreaks\\
&\Xi_{3a}=t\sum_{i=1}^3\frac{\partial}{\partial x_i},\qquad \phi_{(0)}=-\sum_{i=1}^2 m_i x_{(0)i},\quad\phi_{(1)}=-\sum_{i=1}^2 m_i x_{(1)i}-m_3 x_{(0)3};\allowdisplaybreaks\\
&\Xi_{3b}=t\sum_{i=1}^3\frac{\partial}{\partial y_i},\qquad \phi_{(0)}=-\sum_{i=1}^2 m_i y_{(0)i},\quad\phi_{(1)}=-\sum_{i=1}^2 m_i y_{(1)i}-m_3 y_{(0)3};\allowdisplaybreaks\\
&\Xi_{4}=\sum_{i=1}^3\left(\left(y_{(0)i}+\varepsilon y_{(1)i}\right)\frac{\partial}{\partial x_i}-\left(x_{(0)i}+\varepsilon x_{(1)i}\right)\frac{\partial}{\partial y_i}\right),\qquad \phi_{(0)}=\phi_{(1)}=0;\allowdisplaybreaks\\
&\Xi_5=\sum_{j=1}^3\sum_{i=1}^2\left(m_i y_{(0)i}+\varepsilon m_i y_{(1)i}\right)\frac{\partial}{\partial x_j}+\varepsilon m_3y_{(0)3}\sum_{j=1}^2\frac{\partial}{\partial x_j}\allowdisplaybreaks\\
&\quad-\sum_{j=1}^3\sum_{i=1}^2\left(m_i x_{(0)i}+\varepsilon m_i x_{(1)i}\right)\frac{\partial}{\partial y_j}-\varepsilon m_3x_{(0)3}\sum_{j=1}^2\frac{\partial}{\partial y_j}, \qquad \phi_{(0)}=\phi_{(1)}=0;\allowdisplaybreaks\\
&\Xi_{6}=\varepsilon \left(\left(y_{(0)2}-y_{(0)1}\right)\left(m_2\frac{\partial}{\partial x_1}-m_1\frac{\partial}{\partial x_2}\right)\right.\allowdisplaybreaks\\
&\qquad+\left.\left(x_{(0)1}-x_{(0)2}\right)\left(m_2\frac{\partial}{\partial y_1}-m_1\frac{\partial}{\partial y_2}\right)\right),\qquad \phi_{(0)}=\phi_{(1)}=0;\allowdisplaybreaks\\
&\Xi_7=\varepsilon\Xi_1,\qquad \phi_{(0)}=\phi_{(1)}=0;\allowdisplaybreaks\\
&\Xi_{8a}=\varepsilon\Xi_{2a},\qquad\Xi_{8b}=\varepsilon\Xi_{2b},\qquad \phi_{(0)}=\phi_{(1)}=0;\allowdisplaybreaks\\
&\Xi_{9a}=\varepsilon t\sum_{i=1}^2\frac{\partial}{\partial x_i},\qquad \phi_{(0)}=0,\quad\phi_{(1)}=-\sum_{i=1}^2m_i x_{(0)i};\allowdisplaybreaks\\
&\Xi_{9b}=\varepsilon t\sum_{i=1}^2\frac{\partial}{\partial y_i},\qquad\phi_{(0)}=0,\quad\phi_{(1)}=-\sum_{i=1}^2m_i y_{(0)i};\allowdisplaybreaks\\
&\Xi_{10}=\varepsilon\sum_{i=1}^2\left(y_{(0)i}\frac{\partial}{\partial x_i}- x_{(0)i}\frac{\partial}{\partial y_i}\right),\qquad \phi_{(0)}=\phi_{(1)}=0;\allowdisplaybreaks\\
&\Xi_{11}=\varepsilon\frac{\partial}{\partial x_3},\qquad \Xi_{12}=\varepsilon\frac{\partial}{\partial y_3},\qquad \phi_{(0)}=\phi_{(1)}=0.
\end{align*}
As far as the non trivial conservation laws are concerned, the results that we obtain are as~follows:
\begin{itemize}
\item from generator $\Xi_1$ we have the first order approximate conservation of total energy
\begin{align*}
\mathcal{I}_1&=\frac12\left(m_1\dot{\mathbf{r}}_{(0)1}^2+m_2\dot{\mathbf{r}}_{(0)2}^2\right)-\frac{G m_1 m_2}{|\mathbf{r}_{(0)12}|}\\
&+\varepsilon\left(\frac12 m_3\dot{\mathbf{r}}_{(0)3}^2+m_1\dot{\mathbf{r}}_{(0)1}\cdot \dot{\mathbf{r}}_{(1)1}+
m_2\dot{\mathbf{r}}_{(0)2}\cdot \dot{\mathbf{r}}_{(1)2}\right.\\
&\left.-\frac{G m_1 m_3}{|\mathbf{r}_{(0)13}|}-\frac{G m_2 m_3}{|\mathbf{r}_{(0)23}|}+\frac{G m_1 m_2}{|\mathbf{r}_{(0)23}|^3}\mathbf{r}_{(0)12}\cdot\mathbf{r}_{(1)12}\right);
\end{align*}
\item from generators $\Xi_{2a}$ and $\Xi_{2b}$, we have the approximate conservation of total linear momentum
\[
\mathcal{I}_2=m_1\dot{\mathbf{r}}_{(0)1}+m_2\dot{\mathbf{r}}_{(0)2}+\varepsilon\left(m_1\dot{\mathbf{r}}_{(1)1}+m_2\dot{\mathbf{r}}_{(1)2}+m_3\dot{\mathbf{r}}_{(0)3}\right);
\]
\item from the generators $\Xi_{3a}$ and $\Xi_{3b}$ we have
\[
\begin{aligned}
\mathcal{I}_3&=m_1(t\dot{\mathbf{r}}_{(0)1}-\mathbf{r}_{(0)1})+m_2(t\dot{\mathbf{r}}_{(0)2}-\mathbf{r}_{(0)2})\\
&+\varepsilon\left(m_1(t\dot{\mathbf{r}}_{(1)1}-\mathbf{r}_{(1)1})+m_2(t\dot{\mathbf{r}}_{(1)2}-\mathbf{r}_{(1)2})+m_3(t\dot{\mathbf{r}}_{(0)3}-\mathbf{r}_{(0)3})\right),
\end{aligned}
\]
expressing that the approximate (at order $\varepsilon$) barycenter of the system has a uniform and rectilinear motion;
\item from the generator $\Xi_{4}$ we have the approximate conservation of total angular momentum
\[
\begin{aligned}
\mathcal{I}_4&=m_1\mathbf{r}_{(0)1}\wedge \dot{\mathbf{r}}_{(0)1} +m_2 \mathbf{r}_{(0)2}\wedge \dot{\mathbf{r}}_{(0)2}\\
&+\varepsilon\left(m_1\left(\mathbf{r}_{(0)1}\wedge \dot{\mathbf{r}}_{(1)1}+\mathbf{r}_{(1)1}\wedge \dot{\mathbf{r}}_{(0)1}\right)+m_2\left(\mathbf{r}_{(0)2}\wedge \dot{\mathbf{r}}_{(1)2}+\mathbf{r}_{(1)2}\wedge \dot{\mathbf{r}}_{(0)2}\right)\right.\\
&\left.+m_3 \mathbf{r}_{(0)3}\wedge \dot{\mathbf{r}}_{(0)3}\right).
\end{aligned}
\]
\end{itemize}

Therefore, we recovered the approximate conservation laws we expect. The remaining generators do not yield new independent conserved quantities. In fact:
\begin{itemize}
\item the generator $\Xi_{5}$ produces
\[
\begin{aligned}
\mathcal{I}_5&=\left(m_1\mathbf{r}_{(0)1}+m_2\mathbf{r}_{(0)2}\right)\wedge \left(m_1\dot{\mathbf{r}}_{(0)1}+m_2\dot{\mathbf{r}}_{(0)2}+\varepsilon\left(m_1\dot{\mathbf{r}}_{(1)1}+m_2\dot{\mathbf{r}}_{(1)2}+m_3\dot{\mathbf{r}}_{(0)3}\right)\right)\\
&+\varepsilon\left(m_1\mathbf{r}_{(1)1}+m_2\mathbf{r}_{(1)2}+m_3\mathbf{r}_{(0)3}\right)\wedge
\left(m_1\dot{\mathbf{r}}_{(0)1}+m_2\dot{\mathbf{r}}_{(0)2}\right)=\\
&=\left(\int \mathcal{I}_2 dt\right) \wedge \mathcal{I}_2;
\end{aligned}
\]
\item the generator $\Xi_{6}$ produces 
\[
\mathcal{I}_6=\varepsilon\left(\mathbf{r}_{(0)1}-\mathbf{r}_{(0)2}\right)\wedge \left(\dot{\mathbf{r}}_{(0)1}-\dot{\mathbf{r}}_{(0)2)}\right) = \varepsilon \frac{(m_1 + m_2) \mathcal{I}_4- \mathcal{I}_5}{m_1m_2};
\]
\item the generator $\Xi_7$ produces  $\varepsilon\mathcal{I}_1$; 

\item the generators $\Xi_{8a}$ and $\Xi_{8b}$ produce $\varepsilon\mathcal{I}_2$;

\item the generators $\Xi_{9a}$ and $\Xi_{9b}$ produce  $\varepsilon\mathcal{I}_3$; 

\item the generator $\Xi_{10}$ produces $\varepsilon \mathcal{I}_5$;

\item the generators $\Xi_{11}$ and $\Xi_{12}$ produce only trivial conserved quantities. 
\end{itemize}

\section{Concluding Remarks}
\label{sec:conclusions}
In this paper, we used the recently introduced method for determining approximate Lie symmetries of differential equations~\cite{DSGO:lieapprox} in order to state an approximate Noether theorem and 
explicitly construct conservation laws in the cases where the Lagrangian function involves terms with
a small order of magnitude. The approach to approximate Lie symmetries of differential equations proposed in~\cite{DSGO:lieapprox} is consistent with the classical principles of perturbation analysis
\cite{Nayfeh}, and has some advantages if compared with the BGI and FS approaches.
The approximate Noether theorem has been stated for Lagrangian functions depending on $n$ independent variables, $m$ dependent variables and first order derivatives of the latter with respect to the former, and the structure of approximate conservation laws in terms of the generators of approximate variational Lie symmetries and the Lagrangian itself has been derived. We observe that a generalization of such an approximate Noether theorem to higher order Lagrangians can be immediately deduced. All the needed computations have been done by means of the program ReLie~\cite{Oliveri-relie}.  

The consistent approach to approximate variational Lie symmetries has been illustrated through some ordinary differential equations arising from a Lagrangian function: a nonlinearly perturbed harmonic oscillator, a system of two coupled second order ordinary differential equations, and the three-body planar restricted problem without neglecting the gravitational attraction of the smallest body on the main bodies. 

Work is in progress about the application of the present approximate Noether theorem to higher order Lagrangian functions occurring in field theory, as well as to use our approach to approximate Lie symmetries~\cite{DSGO:lieapprox} with the method of partial Lagrangians~\cite{Kara-Mahomed-2006,Kara-Mahomed-Naeem-Wafo-2007,Johnpillai-Kara-Mahomed-2009}.

\section*{Acknowledgments}
Work supported by G.N.F.M. of ``Istituto Nazionale di Alta Matematica''. 
M.G. acknowledges the support through the ``Progetto Giovani No. GNFM 2020''.

\end{document}